\documentclass[conference,10pt]{IEEEtran}

\makeatletter
\def\ps@headings{%
\def\@oddhead{\mbox{}\scriptsize\rightmark \hfil \thepage}%
\def\@evenhead{\scriptsize\thepage \hfil \leftmark\mbox{}}%
\def\@oddfoot{}%
\def\@evenfoot{}}
\makeatother

\usepackage{amssymb,amsmath}
\usepackage{paralist}
\usepackage{graphicx,epsfig}

\usepackage{color}

\newcommand{\F}{\mathbf{F}}

\newcommand{\C}{\mathcal{C}}
\newcommand{\N}{\mathcal{N}}

\newtheorem{theorem}{\textbf{Theorem}}
\newtheorem{lemma}[theorem]{\textbf{Lemma}}
\newtheorem{definition}[theorem]{\textbf{Definition}}

\newcommand{\nix}[1]{}

\begin{document}
\title{Protection Schemes for Two Link Failures\\ in Optical Networks\thanks{This
research was supported in part by grants CNS-0626741 and
CNS-0721453 from the National Science Foundation,
and a gift from Cisco Systems.}}
\author{
\authorblockN{Salah A. Aly ~~~~ and ~~~~~ Ahmed E. Kamal}
\authorblockA{Department of Electrical and Computer Engineering \\Iowa State University, Ames, IA 50011, USA\\Email: \{salah,kamal\}@iastate.edu}
\today{} } \maketitle

\begin{abstract}
In this paper we  develop network protection schemes against two link failures in optical networks. The motivation
behind this work is the fact that the majority of all available links in
an optical network suffer from single  and double link failures. In the
proposed network protection schemes, NPS2-I and NPS2-II, we deploy network coding and
reduced capacity  on the working paths to provide  backup protection
paths. In addition, we demonstrate the encoding and decoding aspects of the proposed  schemes.
\end{abstract}

\section{Introduction}\label{sec:intro}

One of the main services of operation networks that must be deployed
efficiently is reliability. In order to deploy a  reliable networking
strategy, one needs to protect the transmitted signals  over unreliable
links. Link failures are common problems that might occur frequently in
single and multiple operating communication circuits. In network
survivability and network resilience, one needs to design  efficient
strategies to overcome this dilemma.  Optical network
survivability techniques are classified as \textit{pre-designed
protection} and \textit{dynamic restoration}~\cite{zhou00,kamal08a}. The
approach of using pre-designed protection aims to reserve enough
bandwidth such that when a failure occurs, backup paths are used to
reroute the transmission and be able to recover the data. Examples of this
type are 1-1 and 1-N protections~\cite{aly08i,kamal08b}. In dynamic
restoration reactive strategies, capacity is not reserved. However, when the failure occurs, dynamic recovery is used to recover the data transmitted in the links that are suffered from failures. This technique does not  need preserved resources or provision of extra paths that work in case of failure. In this work we will provide several strategies of
dynamic restoration based on coding and reduced distributed fairness
capacities.

Network coding is  a powerful tool that has been recently used to increase
the throughput, capacity, and performance of communication networks.
Information theoretic aspects of network coding have been investigated
in~\cite{soljanin07,ahlswede00}.  Network coding allows the intermediate
nodes not only to forward packets using network scheduling algorithms, but
also encode/decode them using algebraic primitive operations,
see~\cite{ahlswede00,fragouli06,soljanin07} and references therein. As an
application of network coding, data loss because of failures in
communication links can be detected and recovered if the sources are
allowed to perform network coding operations. Network coding is used to
maximize the throughput~\cite{ahlswede00,koetter03}. Also, it is robust
against packet losses and network failures~\cite{ho03,lun05,jaggi07}.

Network protection against  single and multiple link failures  using
adding extra protection paths has been introduced in~\cite{kamal06a,kamal08a}. The
source nodes are able to combine their data into extra paths (backup
\emph{protection paths}) that are used to protect all signals  on the
\emph{working paths} plain carrying data from all sources. In both cases,
$p$-cycles has been  used for protection against single and multiple
link failures.

In this paper we  design two schemes for network protection against  one and
two failed links in a network with $n$ disjoint \emph{working paths}: NPS2-I and NPS2-II. The approach is based on network coding   data
originated by the sources.    We assume that network capacity will be
reduced by a partial factor in order to achieve the required protection.
Several advantages of NPS2-I and NPS2-II can be stated as:
\begin{itemize}
\item

The data sent from the sources are protected without adding extra
paths. Two paths out of the $n$ disjoint \emph{working paths} will carry encoded
data, and hence they are \emph{protection paths}.

\item The encoding and decoding operations are achieved with less
    computational cost at both the sources and receivers. The recovery
    from failures is achieved immediately without asking the senders
    to retransmit the lost data.

\item The normalized network capacity is $(n-2)/n$, which is
    near-optimal in case of using a large number of connections.

\end{itemize}

\section{Network Model}\label{sec:previouswork}
In this section we present the network model and  basic terminology.
\begin{compactenum}[i)]
\item Let $\N$ be a network represented by an abstract graph
    $G=(\textbf{V},E)$, where $\textbf{V}$ is a set of nodes and $E$
    be   a set of undirected edges. Let $S$ and $R$ be  sets of independent sources
    and destinations, respectively. The set $\textbf{V}=V\cup S \cup
    R$ contains the relay, source, and destination nodes. Assume for simplicity that $|S|=|R|=n$,
    hence the set of sources is equal to the set of receivers.

\item  A path (connection) is a set of edges connected together with a
    starting node (sender) and an ending node (receiver).  $$ L_i=
    \{(s_i,e_{1i}),(e_{1i},e_{2i}),\ldots,(e_{(m)i},r_i) \},$$ where
    $1\leq i\leq n$ and $(e_{(j-1)i},e_{ji}) \in E$ for some integer $m$.
\item

The node can be a router, switch, or an end terminal depending on the
network model $\N$ and the transmission layer.

\item

$L$ is a set of paths $L=\{L_1,L_2,\ldots,L_n\}$ carrying the data from
the sources to the receivers as shown in Fig.~\ref{fig:nnodes}. All
connections have the same bandwidth, otherwise a connection with high
bandwidth can be divided into multiple connections, each of which has a
unit capacity. There are exactly $n$ connections. For simplicity, we
assume that the number of sources is less than or equal to the number of
available paths. A sender with a high capacity can divide its capacity
into multiple unit capacities, each of which has its own path.
  The failure on a link $L_i$ may occur due to the network
    circumstance such as a link replacement, overhead, etc. We do not address  in this work  the failure cause. However, we
    assume that there are one or two failures in the working paths and the
    protection strategy is able to protect/recover it.
\item Each sender $s_i \in S$ will transmit its own data $x_i$ to a
    receiver $r_i$ through a connection $L_i$. Also, $s_i$ will transmit
    encoded data $\sum_{i}^n x_i$ to $r_i$ at a different time slot if it is
    assigned to send the encoded data.

\item The data from all sources are sent in sessions. Each session has a
    number of time slots $n$. Hence $t_\delta^\ell$ is a value at round time
    slot $\ell$ in session $\delta$.

\item In this model $\N$ if we consider only a single link failure, then
    it    is     sufficient to apply the encoding and decoding operations over a finite
    field with two elements,  $\F_2=\{0,1\}$. However, if there are
    double failures, then a finite field with higher alphabets is
    required.

\item There are at least two receivers and two senders with at least two
    disjoint paths. Otherwise,  the protection strategy can not be deployed for a single
    path, which it can not protect itself.
\end{compactenum}

\begin{figure}[t]
 \begin{center} 
  \includegraphics[scale=0.8]{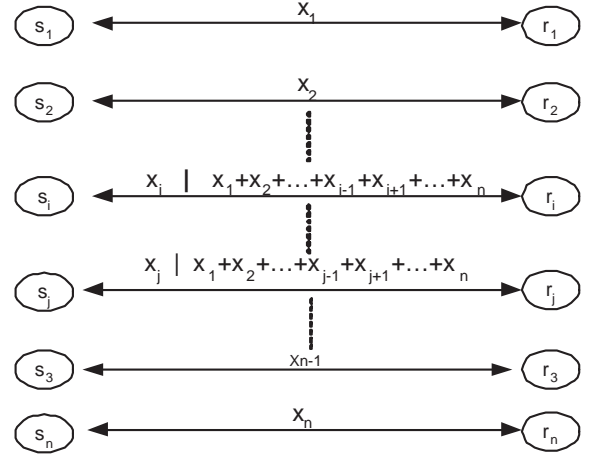}
  \caption{Network protection against  link/path failures using reduced capacity and network coding. Two paths out of  $n$ disjoint \emph{working  paths} carry encoded data for protection against two link failures.}\label{fig:nnodes}
\end{center}
\end{figure}

We can define the network capacity in the light of min-cut max-flow
information theoretic view~\cite{ahlswede00}.
\begin{definition}\label{def:capacitylink}
The capacity of a connecting path $L_i$ between $s_i$ and $r_i$ is defined
by
\begin{eqnarray} c_i=\left\{
      \begin{array}{ll}
        1, & \hbox{$L_i$ is \emph{active};} \\
        0, & \hbox{otherwise.}
      \end{array}
    \right.
\end{eqnarray}
The total capacity is given by the summation of all paths' capacities.
What we mean by an \emph{active}  link is that the receiver is able to
receive and process signals/messages throughout this link.
\end{definition}

 	\smallbreak

Clearly, if all links are active then the total capacity is $n$ and
normalized capacity is $1$. In general the normalized capacity of the
network for the active and failed links is computed as:
\begin{eqnarray}
C_\N=\frac{1}{n}\sum_{i=1}^n c_i.
\end{eqnarray}

We  define the \emph{working paths} and \emph{protection paths} as
follows:

\begin{definition}
The\emph{ working paths} of a network with $n$ connection paths are the
paths that carry unencoded data traffic  under normal operations. The
\emph{Protection paths} are paths that provide alternate backup paths to carry
encoded data traffic in case of failures. A protection scheme ensures that
data sent from the sources will reach the receivers in case of failure
incidences on the working paths.
\end{definition}

%


Every sender $s_i$ prepares a packet \emph{$packet_{s_i \rightarrow r_i}$}
 to send to the receiver $r_i$. The packet contains the sender's ID, data
$x_i^\ell$, and a round time for every session $t^\ell_\delta$ for some
integers $\delta$ and $\ell$. We have two types of packets:
\begin{compactenum}[i)]
\item Packets sent without coding, in which the sender does not need
    to perform any coding operations. For example, in case of packets
    sent without coding, the sender $s_i$ sends the following packet
    to the receiver $r_i$.
\begin{eqnarray}
packet_{s_i \rightarrow r_i}:=(ID_{s_i},x_i^\ell,t_\delta^\ell)
\end{eqnarray}

\item  Packets sent with encoded data, in which the sender needs to
    send other sender's data. In this case, the sender $s_i$ sends the following packet to receiver
$r_i$:
\begin{eqnarray}
packet_{s_i \rightarrow r_i}:=(ID_{s_i},\sum_{j=1,j\neq i}^n x_j^\ell,t^\ell_\delta).
\end{eqnarray}
The value $y_i^{\ell}=\sum_{j=1,j\neq i}^n x_j^\ell$  is computed by every
sender $s_i$, by  collecting the data from all other
senders and encoding them using the bit-wise operation.
\end{compactenum} In either case, the sender has a full capacity in the
connection path $L_i$.

 \smallbreak

The protection path that carries the encoded data from all sources is used
for the data recovery in case of failure. Assuming  the encoding
operations occur in the same round time of a particular session, every
source $s_i$ adds its value, for $1 \leq i \leq n$. Therefore, the encoded
data over the protection path is $y_i=\sum_{j=1, i\neq j}^n x_j$. The
decoding operations are done at every receiver $r_i$ by adding the data
$x_i$ received over the working path $L_i$. The node $r_k$ with failed
connection $L_k$ will be able to recover the data $x_k$. Assuming all
operations are achieved over the binary finite field $\F_2$. Hence we have
\begin{eqnarray}x_k= y_i-\sum_{j=1, i \neq j}^n x_j.\end{eqnarray}

\section{Protections Using Dedicated Paths (NPS2-I)}
In this section we develop a network protection scheme (NPS2-I) for two link failures in optical networks. The protection scheme is achieved using network coding and dedicated paths.  Assume we have $n$ connections carrying data from a set of $n$ sources to a set of $n$ receivers. All connections represent disjoint paths, and the sources are independent of each other. The authors in~\cite{kamal06a,aly08preprint1} introduced a model for
optical network protection against a single link failure using an extra and dedicated paths provisioning. In this model NPS2-I  we extend this approach to two link failures.

We will  provide two backup paths to protect against any two disjoint links,
which might experience failures. These two protection paths can be chosen
using network provisioning. The protection paths are fixed for all rounds
per session, but they may vary among sessions. For example, sender $s_i$  transmits a message $x_i^\ell$ to  a receiver $r_i$ at time $t_\delta^\ell$ in round time $\ell$ in session $\delta$. This process is explained
in Scheme~(\ref{eq:n-1protection2}) as:

\begin{eqnarray}\label{eq:n-1protection2}
\begin{array}{|c|ccccc|c|}
\hline
& \multicolumn{5}{|c|}{\mbox{ round time session 1 }}&\ldots    \\
\hline
&1&2&3&\ldots&n&\!\!\ldots\\
\hline    \hline
  s_1 \rightarrow r_1 & x_1^1&x_1^2 &x_1^3&\ldots  &x_1^n &\ldots \\
    s_2 \rightarrow r_2 &  x_2^1&  x_2^2&x_2^3&\ldots&x_2^{n} &\ldots \\
s_3 \rightarrow r_3 &  x_3^1& x_3^2& x_3^3&\ldots&x_3^{n} &\ldots \\
 \vdots&\vdots&\vdots&&\vdots&\vdots&\ldots\\
    s_i \rightarrow r_i& x_i^1 & x_i^2 &x_i^3&\ldots& x_i^{n}&\ldots\\
     s_j \rightarrow r_j & y_j^1&y_j^2&y_j^3&\ldots&y_j^{n}&\ldots\\
     s_{k} \rightarrow r_{k} & y_{k}^1&y_{k}^2&y_{k}^3&\ldots&y_{k}^{n}&\ldots\\
     s_{i+1} \rightarrow r_{i+1}& x_{i+1}^1 & x_{i+1}^2 &x_{i+1}^3&\ldots& x_{i+1}^{n}&\ldots\\
 \vdots&\vdots&\vdots&\vdots&\vdots&\vdots&\ldots\\
   s_n \rightarrow r_n & x_n^1&x_n^2&x_n^3&\ldots&x_{n}^n&\ldots\\
\hline
\hline
\end{array}
\end{eqnarray}
All $y_j^\ell$'s are defined as:
\begin{eqnarray} y^\ell_j=\sum_{i=1,i\neq j \neq k}^n a_i^\ell x_i^\ell \mbox{  and  } y^\ell_k=\sum_{i=1,i\neq k \neq j}^n b_i^\ell x_i^\ell.
\end{eqnarray}

The coefficients $a_i^\ell$ and $b_i^\ell$ are chosen over a finite field
$\F_q$ with $q > n-2$, see~\cite{aly08preprint1} for more details. One way to choose these coefficients is by using the follow two vectors.
\begin{eqnarray}\label{eq:twovectors}
\left[\begin{array}{ccccc}
1&1&1&\ldots&1\\
1&\alpha&\alpha^2&\ldots&\alpha^{n-3}
\end{array}\right]
\end{eqnarray}
Therefore, the coded data is
\begin{eqnarray} y^\ell_j=\sum_{i=1,i\neq j \neq k}^n  x_i^\ell \mbox{  and  } y^\ell_k=\sum_{i=1,i\neq k \neq j}^n \alpha^{i \mod n-2} x_i^\ell.
\end{eqnarray}
In the case of two failures, the receivers will be able to solve two linearly independent equations in two unknown variables. For instance, assume the two failures occur in paths number two and four. Then the receivers will be able to construct two equations with cofficients
\begin{eqnarray}\label{eq:twovectors2}
\left[\begin{array}{cc}
1&1\\
\alpha&\alpha^{3}
\end{array}\right]
\end{eqnarray}
 Therefore, we have
\begin{eqnarray}
x_2^\ell+x_4^\ell\\
\alpha x_2^\ell+\alpha^3 x_4^\ell
\end{eqnarray}
One can multiply the first equation by $\alpha$ and subtract the two equations to obtain value of $x_4^\ell$.

We  notice that the encoded data symbols
$y_j^\ell$ and $y_k^\ell$ are fixed per one session transmission but it is
varied for other sessions. This means that the path $L_j$ is dedicated to
send all encoded data $y_j^1,y_j^2,\ldots,y_j^n$.
\begin{lemma}
The normalized capacity of NPS2-I of the network model $\N$ described
in~(\ref{eq:n-1protection2}) is given by
\begin{eqnarray}
\C=(n-2)/n.
\end{eqnarray}
\end{lemma}
\begin{proof}
There are $n$ rounds in every session. Also, we have $n$ connections per a
round time. There exist two connections which carry backup data for
protection, hence there are $n-2$ connections that carry working data.
Therefore, the normalized capacity is given as:
\begin{eqnarray*}
\C=(n-2)n/n^2,
\end{eqnarray*}
which gives the result.
\end{proof}

In NPS2-I there are three different scenarios for two link failures, which
can be described as follows:
\begin{compactenum}[i)]
\item If the two link failures occur in the backup protection paths $L_j$
    and
    $L_k$,
    then no recovery operations are required at the receivers side. The
    reason is that these two paths are used for protections, and all other
    working paths will convey the data from the senders to  receivers.

\item If the two link failures occur in one backup protection path say
    $L_j$ and one working path $L_i$, then  recovery operations are
    required. The receiver $r_i$ must recover its data using one of the
    protection paths.
\item If the two link failures occur in two working paths, then in this
    case the two protection paths are used to recover the lost data. The
    idea of recovery in this case is to build a system of two equations
    with two unknown variables.
\end{compactenum}

\section{Protection Against Two Link Failures (NPS2-II)}\label{sec:singlefailure}
In this section we will provide an approach for network protection against
two link failures in optical networks.  We deploy network coding and
distribute capacity over the \emph{working paths}.  We will compute the
network capacity in this approach. In~\cite{aly08preprint1} we will
illustrate the tradeoff  and implementation aspects of  this approach,
where there is enough space for details.

We assume that there is a set of $n$ connections that need to be protected
with $\%100$ guarantee against  single and two link failures. Assume
$\F_q$ is a finite field with $q$ elements. For simplicity we consider n
is an even number.

\subsection{Two Link Failures, Achieving $(n-2)/n$ Capacity}
Let $x_i^\ell$ be the data sent from the source $s_i$ at round time $\ell$
in a session at time $t_\delta^\ell$. Also, assume the encoded data
$y_i=\sum_{j=1,j \neq i}^n x_j^\ell$. Put differently:
\begin{eqnarray}
y_i=x_1^\ell\oplus x_2^\ell\oplus \ldots \oplus x_{i \neq j}^\ell\oplus \ldots\oplus x_n^\ell.
\end{eqnarray}
The protection scheme NPS2-II runs in sessions as explained below. Every
session has at most one single failure through out its each round time. As
shown in Scheme~(\ref{eq:n-1protection}), the protection matrix for the
first session is given by the following protection code:

\begin{eqnarray}\label{eq:n-1protection}
\begin{array}{|c|cccccc|}
\hline \hline
& \multicolumn{6}{|c|}{\mbox{ round time session 1 }} \\
\hline
&1&2&3&\ldots&\ldots&\ell\\
\hline    \hline
  s_1 \rightarrow r_1 & y_1^1&x_1^1&x_1^2 &\ldots &\ldots  &x_1^{\ell-1}  \\
    s_2 \rightarrow r_2 &  y_2^1&x_2^1&  x_2^2&\ldots&\ldots&x_2^{\ell-1} \\
s_3 \rightarrow r_3 &  x_3^1&y_3^2 &x_3^2&\ldots&\ldots&x_3^{\ell-1}  \\
s_4 \rightarrow r_4 &  x_4^1&y_4^2& x_4^2&\ldots&\ldots&x_3^{\ell-1} \\
 \vdots&\vdots&\vdots&\vdots&\vdots&\vdots&\vdots\\
    s_i \rightarrow r_i& x_i^1 & x_i^2 &\ldots&y_{i}^j &\ldots& x_i^{\ell-1}\\
    s_{i+1} \rightarrow r_{i+1}& x_{i+1}^1 & x_{i+1}^2 &\ldots&y_{i+1}^j &\ldots& x_{i+1}^{\ell-1}\\
 \vdots&\vdots&\vdots&\vdots&\vdots&\vdots&\vdots\\
   s_{n-1} \rightarrow r_{n-1} & x_{n-1}^1&x_{n-1}^2&x_{n-1}^3&\ldots&\ldots&y_{n-1}^\ell\\
   s_n \rightarrow r_n & x_n^1&x_n^2&x_n^3&\ldots&\ldots&y_{n}^\ell\\
\hline
\end{array}
\end{eqnarray}
where
\begin{eqnarray}\label{eq:y_NPS-T}
y_k^\ell=\sum_{i=1}^{2(\ell-1)} a_i^{\ell-1} x_i^{\ell-1} + \sum_{i=2\ell+1}^n a_i^\ell x_i^\ell  \nonumber \\ \mbox{   for  } (2\ell-1) \leq k \leq 2\ell,  1 \leq \ell \leq n/2.
\end{eqnarray}

All coefficients are taken from $\F_q$ for $q> n-2$, see~\cite{aly08preprint1} for more details. Also, the two vectors shown in~\ref{eq:twovectors} can be used in this case.
We note that the
data symbols in NPS2-II are sent in independent sessions. This means that
every session has its own recovery scheme. Also, two failures occur in
only two connections during the session round times. Hence the sender
$s_i$ sends the message $x_i^j$ for all $1\leq j \leq \ell-1$ and $1 \leq
i \leq n$ during the first session. One can always change the round time
of the encoded data $y_k^\ell$ and the data $x_i^{j}$ for any round time
$j$ in the same session.


Now, we shall compute the normalized capacity of NPS2-II for the network
$\N$ at one particular session; the first session. The capacity is
calculated using the well-known min-cut max-flow
theorem~\cite{ahlswede00}.
\begin{theorem}\label{thm:n-1capacity}
The optimal fairness distributed normalized capacity of NPS2-II shown in
Scheme~(\ref{eq:n-1protection}) is given by
\begin{eqnarray}
\C=(n-2)/n.
\end{eqnarray}
\end{theorem}
\begin{proof}
Let $n$ be the number of sources, each of which has a unit capacity in the
connection $L_i$ from $s_i$ to $r_i$. Let $j$ be an index of an arbitrary
session that has two link failures. We have $n$ paths that have capacity
$n$.  Also, we have $\ell=n/2$ round times, in which each round time has
$n-2$ capacity in the working paths. Therefore the total capacity is given
by
\begin{eqnarray}
(n-2)(\ell)=(n^2-2n)/2.
\end{eqnarray}
By normalizing this value with the total capacity $n\ell$, then the result
follows.
\end{proof}

The network protection strategy NPS2-II against one or two link failures
is deployed in two processes: Encoding and decoding operations. The
encoding operations are performed at the set of sources, in which one or
two sources will send the encoded data depending on the number of
failures. The decoding operations are performed at the receivers' side, in
which  receivers with  failed links had to receive all other receivers'
data in order to recover their own data. Depending on NPS2-I or NPS2-II
the receivers will experience some delay before they can actually decode
the packets. The transmission is done in rounds, hence linear
combinations of data has to be from the same round time. This can be
achieved using the round time that is included in each packet sent by a
sender.

Assume there are data collectors $\mathcal{S}$ and $\mathcal{R}$ at the
senders and receivers, respectively. They can be a sender (receiver) node
to send (receiver) encoded data, see~\cite{aly08preprint1}.

\bigskip

 \noindent \textbf{Encoding Process:} The encoding operations are for
 each round per a session.
\begin{itemize}
\item

The source nodes send a copy of their data to the data distributor
$\mathcal{S}$, then $\mathcal{S}$ will decide which source will send
the encoding data $y_k^\ell$ and all other sources will send their own
data $x_i^\ell$. This process will happen in every round during
transmission time.
\item The encoding is done by linear  operation of sources'
    coefficients which is the fastest arithmetic operation that can be
    performed among all sources' data.
\item The server $\mathcal{S}$ will change the sender $s_i$ that
    should send the encoded data $y_i^\ell$ in every round of a given
    session for the purpose of fairness and distributed capacities.
\end{itemize}

\begin{figure}[t]
 \begin{center} 
  \includegraphics[scale=0.5]{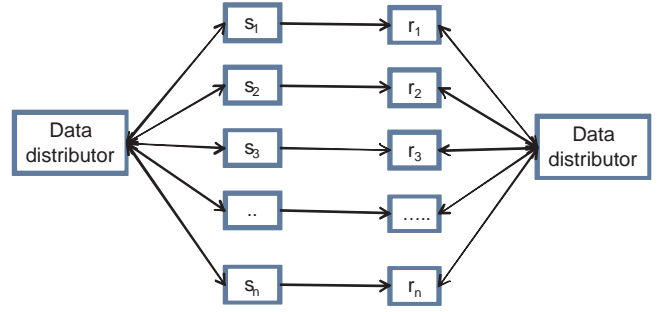}
  \caption{The encoding and decoding operations are done at the data distributor and collector, respectively.}\label{fig:nnodes}
\end{center}
\end{figure}

\bigskip

\noindent \textbf{Decoding Process:} The objective of the decoding and
recovery process is to withhold rerouting the signals or the transmitted
packets due to link failures,
see~\cite{aly08i,aly08preprint1,kamal08a}. { However, we provide
strategies that utilize network coding and reduced capacity at the source
nodes. We assume that the receiver nodes are able to perform decoding
operations using a data collector $\mathcal{R}$.}

We assume there is a data distributor $\mathcal{R}$ that will collect data
from all \emph{working} and \emph{protection paths} and is able to perform the decoding
operations.   In this case we assume that all receivers $R$ have available
shared   paths with the data collector $\mathcal{R}$. At the receivers
side, if there are  two failures in  paths $L_j$ and $L_k$, then there are
several situations.

\begin{itemize}
\item If the paths $L_j$ and $L_k$ carry unencoded data  (\emph{working paths}), then the data distributor $\mathcal{R}$ must query
    all other nodes in order to recover the failed data. In this case
    $r_k$ and $r_j$ must query $\mathcal{R}$ to retrieve their lost
    data.

\item If the path $L_k$ carries encoded data $y_k$ (\emph{protection path})
    and $L_j$ carries unencoded data (\emph{working path}), then data
    collector $\mathcal{R}$ must query all other receivers in order to
    perform decoding, and $r_j$ receives the lost data  $x_j^\ell$.
\item If the paths $L_j$     and $L_k$ carry encoded data (they are
    both \emph{protection paths}), then no action is required.
    \end{itemize}

\bigskip

\section{Conclusion}\label{sec:conclusion}
In this paper we presented network protection schemes NPS2-I and NPS2-II
against  single and double  link  failures in optical networks. We showed
that protecting two failures can be achieved using network coding and
reduced capacity. The normalized capacity of the proposed schemes is
$(n-2)/n$, which is near optimal for a large number of connections. Extended version of this paper will appear in~\cite{aly08preprint1}, where protection against $t$ multiple failures is investigated.

\bigskip

\section{ACKNOWLEDGMENTS}
This research was supported in part by grants CNS-0626741 and
CNS-0721453 from the National Science Foundation,
and a gift from Cisco Systems. \emph{S~A.~A would like to thank Prof. A~E.~K for his support and guidance. He is  also grateful to Prof. E.~Soljanin and Prof. A.~Klappenecker for their kindness and for training him how to inscribe high quality research papers.}

\bigskip

\scriptsize
\bibliographystyle{plain}

\end{document}